\documentclass[11pt]{article}
\usepackage[utf8]{inputenc}
\usepackage{fullpage}
\usepackage{amsmath, amssymb}
\usepackage{stmaryrd}
\usepackage{xcolor}
\usepackage{listings}

\usepackage{amsthm}
\usepackage{young}
\usepackage{multicol}
\usepackage{algorithm}%
\usepackage{algpseudocode}%
\usepackage{url}
\MakeRobust{\Call}
\newcommand{\coef}[2]{C_{#1}(#2)}

\usepackage{float}
\usepackage[sorting=nyt]{biblatex}
\usepackage{multicol}
\usepackage{multirow}
\usepackage{graphicx}
\newtheorem{theorem}{Theorem}

\newcommand{\R}{\mathbb{R}}
\newcommand{\C}{\mathbb{C}}
\theoremstyle{definition}
\newtheorem{definition}{Definition}[section]

\newtheorem{lemma}[]{Lemma}
\theoremstyle{plain}
\newtheorem{corollary}{Corollary}

\begin{document}
\begin{titlepage}
\title{%
A Non-iterative Parallelizable Eigenbasis Algorithm for Johnson Graphs \\
}

\author{Jackson Abascal\footnote{Department of Computer Science and Mathematics, University of Rochester, Rochester NY, 14627} \\ \url{jabascal@u.rochester.edu} 
\and
Amadou Bah\footnote{Department of Electrical Engineering and Computer Science, Massachusetts Institute of Technology, Cambridge MA, 98926} \\ \url{abah@mit.edu}
\and
Mario Banuelos\footnote{Department of Mathematics, California State University, Fresno, Fresno CA, 93740} \\ \url{mbanuelos22@csufresno.edu}
\and
David Uminsky\footnote{University of San Francisco, San Francisco CA, 94117} \\ \url{duminsky@usfca.edu}
\and
Olivia Vasquez\footnote{Department of Mathematics, Central Washington University, Ellensburg WA, 98926} \\ \url{vasquezol@cwu.edu} }
\date{\today.}

\maketitle
\begin{abstract}
We present a new $O\left(k^2\binom{n}{k}^2\right)$ method for generating an orthogonal basis of eigenvectors for the Johnson graph $J(n,k)$.
Unlike standard methods for computing a full eigenbasis of sparse symmetric matrices, the algorithm presented here is non-iterative, and produces exact results under an infinite-precision computation model.
In addition, our method is highly parallelizable; given access to unlimited parallel processors, the eigenbasis can be constructed in only $O(n)$ time given $n$ and $k$.
We also present an algorithm for computing projections onto the eigenspaces of $J(n,k)$ in parallel time $O(n)$.
\end{abstract}

\end{titlepage}

\section{Introduction}
The Johnson graph $J(n, k)$ has the $k$-element subsets of $\{1, \dots, n\}$ as vertices.
There is an edge between two subsets $S$ and $T$ exactly when $|S \cap T| = k-1$.
In this work, we present an efficient and highly parallelizable algorithm which computes the orthogonal eigenbasis, as described in \cite{filmus2016orthogonal}, in time $O\left(k^2 \binom{n}{k}^2\right)$. 
Most known eigenvector methods to date are iterative, and result only in numerical solutions \cite{strang2006linear}.
The eigenspaces of Johnson graphs have a number of remarkable combinatorial properties, and have found many uses both in pure and applied mathematics.
Spectral analysis, a technique pioneered by Diaconis in \cite{diaconis1988group}, is a method of analyzing data modeled on a set $X$ with a natural group action $G \to \operatorname{Aut}(X)$. Let the unitary representation $\rho$ be the extension of this action from $G$ to the free vector space on $X$ with coefficients in $\C$.
Then $\rho$ can be broken into a direct sum of irreducible representations, and each $f:X \to \C$ can be decomposed as the set of isotypic projections of $f$ onto the irreducible components.
The Johnson graph becomes important when $X$ is the $k$-element subsets of $\{1,\dots,n\}$ and $G$ is the symmetric group, where $\sigma \in S_n$ sends $S = \{s_1, s_2,\dots,s_k\}$ to $\sigma S = \{\sigma s_1, \sigma s_2, \dots, \sigma s_3\}$.
In this case, the isotypic components are exactly the eigenspaces of $J(n, k)$ \cite{orrison2001eigenspace}. 

This technique was used by Uminsky and Orrison to analyze voting coalitions in the United States Supreme court in \cite{uminsky2003generalized}, \cite{lawson2006spectral}, and by Diaconis in \cite{diaconis1989} for analyzing voting data from the American Psychological Association. 

Maslen, Orrison, and Rockmore present an efficient method for projection of a vector onto eigenspaces of real symmetric matrices in \cite{maslen2003computing} using Lanzcos iteration. 
Other efficient methods for computing such projections have been developed by \cite{diaconis1995}, \cite{kondor2011non}, and \cite{iglesias2017complexity}. 
All of these methods focus exclusively on the projection of a given vector without considering the computation of an eigenbasis. 

The Johnson graph also has important ties to coding theory due to its close relation to the Johnson association scheme.
A more detailed account of how the Johnson scheme can be used in error correcting codes can be found in \cite{delsarte2006}.

In \cite{filmus2016orthogonal}, Filmus gives an explicit description of an orthogonal eigenbasis for the eigenspaces of Johnson graphs. %
Filmus doesn't consider the algorithmic complexity of his formula, however a naive implementation does not lend to any efficient algorithm.
For $J(n, k)$, we are led to evaluate $j$ 
In this work, we derive an alternative constructive description of the same eigenbasis from Filmus's work, leading to an efficient algorithm for generating the eigenbasis.
Our algorithm runs in serial time $O(k^2 \binom{n}{k}^2)$, but leads to a highly parallelizable version which can generate an orthogonal eigenbasis and perform projections onto the eigenbasis in parallel time $O(n)$, assuming an unbounded number of processors.

\section{An Orthogonal Basis for Johnson Matrices} \label{ortho_constr_sec}
We begin with some preliminary definitions, which closely follow the terminology of \cite{filmus2016orthogonal}.
Let $[n]$ denote the set $\{1,\dots,n\}$.
Let $\mathcal{S}_{n,d}$ be the set of sequences of distinct elements in $[n]$ with length $d$.

\begin{definition}
For two sequences $A$ and $B$ in $\mathcal{S}_{n,d}$, we say that $A \prec B$ if (1) the $i$-th element of $A$ is less than the $i$-th element of $B$, and (2) $A$ and $B$ share no elements.
\end{definition}
Note that $\prec$ is a partial order on $\mathcal{S}_{n,d}$.

\begin{definition}
A sequence $B \in \mathcal{S}_{n,d}$ is a \textit{top set} if it is increasing, and there exists some $A$ such that $A \prec B$.
\end{definition} 

Let $\mathcal{B}_{n,d}$ be the set of all top sets in $\mathcal{S}_{n,d}$.
For any $A \prec B$ in $\mathcal{S}_{n, d}$, define the multivariate polynomial
\[\chi_{A,B}(x_1,\dots,x_{n}) = \prod_{i=1}^{d}(x_{A_i} - x_{B_i})\] where $A_i$ and $B_i$ are the $i$-th elements of $A$ and $B$, respectively.
Now define
\[\chi_B = \sum_{A \prec B} \chi_{A,B}.\]
and let $\mathcal{Y}_{n,d} = \{\chi_B : B \in \mathcal{B}_{n,d}\}$.
Let $\binom{[n]}{k}$ be the set of $(x_1,\dots,x_{n}) \in \{0,1\}^n$ such that $\sum_{i=1}^{n} x_i = k$.
Each column of the adjacency matrix of $J(n,k)$ is labeled with a subset $S$, but we could instead label it with the corresponding element $(x_1,\dots, x_{n})$ in $\binom{[n]}{k}$ where $x_i = 1$ if and only if $i \in S$.

There are $k+1$ eigenspaces of $J(n,k)$, with eigenvalues $(k-d)(n-k-d)-d$ for $d \in \{0,1,\dots,k\}$.
These eigenspaces are mutually orthogonal
\cite{delsarte1973algebraic}.
Let $M_d$ be the space of vectors with eigenvalue $(k-d)(n-k-d)-d$.
To construct an orthogonal basis for $M_d$, we construct a basis eigenvector for each $\chi \in \mathcal{Y}_{n,d}$.
Given such a $\chi$, there is an eigenvector whose entry associated with column $(x_1,\dots,x_n) \in \mathcal{Y}_{n,d}$ is equal to $\chi(x_1,\dots,x_{n})$.
We can compute this by setting the eigenvector coefficient in the column associated with $(x_1,\dots, x_n) \in \binom{[n]}{k}$ to $\chi(x_1,\dots,x_{n})$.

\subsection{Example for J(4,2)}
Suppose we want to compute a set of eigenvectors that form an orthogonal basis for $M_{1}$ in $J(4, 2)$.
We begin by constructing $\binom{[4]}{2} = \{(1, 1, 0, 0), (1,0,1,0), (1,0,0,1), \dots\}$.
To find a basis of this space we need to find $\mathcal{Y}_{4,1}$.
Now $\mathcal{B}_{4,1} = \{\{2\},\{3\},\{4\}\}$.
Thus:
\begin{itemize}
\item For $B = \{2\}$,  $A$ can only be $\{1\}$.
\item For $B = \{3\}$,  $A$ can be $\{1\}$ or $\{2\}$.
\item For $B= \{4\}$, $A$ can be $\{1\}$, $\{2\}$, or $\{3\}$.
\end{itemize}
We can now write each $\chi_{B}$ for $B \in \mathcal{B}_{4,2}$ as function of $x_1$, $x_2$, $x_3$, and $x_4$:

\begin{itemize}
\item $\chi_{\{1\}} = (x_1-x_2)$.
\item $\chi_{\{2\}} = (x_1-x_3)+(x_2-x_3)$.
\item $\chi_{\{3\}} = (x_1-x_4)+(x_2-x_4)+(x_3-x_4)$.
\end{itemize}

Evaluating $\chi_{B}$ for each $B$ at all entries in $\binom{[4]}{2}$ produces Table \ref{eigenbasis}, where
each row corresponds to an eigenvector and together the rows form a basis for $M_1$. 

\begin{table}[ht]
\centering
\begin{tabular}{ r|c|c|c|c|c|c| }
\multicolumn{1}{r}{}
 &  \multicolumn{1}{c}{1100}
 & \multicolumn{1}{c}{1010}
  &  \multicolumn{1}{c}{1001}
 & \multicolumn{1}{c}{0110} 
  &  \multicolumn{1}{c}{0101}
 & \multicolumn{1}{c}{0011} \\
\cline{2-7}
$\chi_{\{1\}}$ & $0$ & $1$ & $1$ & $-1$ & $-1$ & $0$ \\
\cline{2-7}
$\chi_{\{2\}}$ & $2$ & $-1$ & $1$ & $-1$ & $1$ & $-2$ \\
\cline{2-7}
$\chi_{\{3\}}$ & $2$ & $2$ & $-2$ & $2$ & $-2$ & $-2$ \\
\cline{2-7}
\end{tabular}
\caption{The eigenbasis of $M_1$ for the $J(4,2)$.}
\label{eigenbasis}
\end{table}

\subsection{Coefficients of \boldmath{$\chi_B$}} 
\begin{definition}
Let $f(x_1,\dots,x_{n})$ be a polynomial, and $S = \{s_1,\dots,s_{k}\}\subseteq [n]$.
We define $\coef{f}{S}$ as the coefficient of $x_{s_1}\cdots x_{s_{k}}$ in $f$.
\end{definition}

\begin{theorem}\label{coef}
Let $S \in [n]$ be of size $k$.
For $i \in [n]$, if $i \in S$, let $x_i = 1$ and let $x_i = 0$ otherwise. Let $B$ be a top set of size $d \leq k$.
Then
\[\chi_B(x_1,\dots,x_{n}) = \sum_{\substack{T \subseteq S \\ |T| = d}} \coef{\chi_B}{T}.\]
\end{theorem}
\begin{proof}
Recall that $\chi_B$ is multilinear and homogeneous of degree $d$.
Let $T = \{t_1,\dots,t_{d}\} \subseteq [n]$.
If $x_{t_i} = 0$ for any $t_i \in T$, then the monomial $x_{t_1}\cdots x_{t_{k}}$ will be evaluated to zero on inputs $x_1,\dots,x_{n}$.
If the monomial is not evaluated to zero, then $x_{t_i} = 1$ for all $i$, and thus each $t_i \in S$.
So $T \subseteq S$.
Since nonzero $x_i$ are 1, it follows that $\chi_B(x_1,\dots,x_{n})$ is exactly the sum of all coefficients whose monomial is nonzero on inputs $x_1,\dots,x_{n}$, which are exactly $\coef{\chi_B}{T}$ for all size $d$ subsets $T$ of $S$.
\end{proof}

\begin{corollary}
Let $S \subseteq [n]$ be of size $k$.
For $i \in [n]$ let $x_i = 1$ if $i \in S$, and $x_i = 0$ otherwise.
Let $B$ be a top set of size $k$.
Then $\chi_{B}(x_1,\dots,x_{n})$ is the coefficient of the monomial in $\chi_B$ with variables indexed by $S$.
\end{corollary}
\begin{proof}
Since $\chi_{B}$ is multilinear and homogeneous of degree $k$, it follows that the only monomial which evaluates to a nonzero number is the monomial indexed by elements in $S$. Furthermore, since the variables indexed by elements in $S$ are evaluated at 1, $\chi_B(x_1,\dots,x_{n})$ is simply the coefficient of the monomial with variables indexed by $S$.
\end{proof}

\begin{definition}
Let $U$ be a set of length $k$ and let $V$ be a sequence of length $k$. Define $U-V$ as the elements in $U$ not in $V$, and $U \cap V$ as the elements in both $U$ and $V$.
\end{definition}

\begin{theorem} \label{sgnCof}
Let $B$ be a top set of size $d$.
Let $S = \{s_1,\dots,s_{d}\}$ be a subset of $[n]$.
Then $\coef{\chi_B}{S}$ is equal to $(-1)^{|S \cap B|}$ multiplied by the number of $A \prec B$ such that $\coef{\chi_{A,B}}{A} \neq 0$.
\end{theorem}
\begin{proof}
Suppose $\coef{\chi_{A,B}}{A}$ is nonzero.
Recall that
\[\chi_{A,B} = (x_{A_1} - x_{B_1})(x_{A_2} - x_{B_2}) \cdots (x_{A_{d}} - x_{B_{d}}).\]

Any monomial in $\chi_{A,B}$ must be of the form $x_{t_1} \cdots x_{t_{d}}$ where $t_i$ is in $\{A_i, B_i\}$.
Let $T = \{t_1,\dots,t_{d}\}$.
The coefficient of this monomial is a product over $1 \leq i \leq d$ of the coefficient of $x_{t_i}$ in $x_{A_i} - x_{B_i}$, and in particular is equal to $1^{|T \cap A|}(-1)^{|T \cap B|}$.
Since $T$ must be a permutation of $S$ for $x_{t_1}\cdots x_{t_{d}}$ to equal $x_{s_1}\cdots x_{s_{d}}$, it follows that $|T \cap B| = |S \cap B|$, and so $\coef{\chi_{A,B}}{S}= (-1)^{|S \cap B|}$.
Thus for all $A \prec B$, it follows that $\coef{\chi_{A,B}}{S}$ is either 0, or $(-1)^{|S \cap B|}$.
Combined with the fact that
\[\coef{\chi_B}{S} = \sum_{A \prec B} \coef{\chi_{A,B}}{S},\]
we obtain the desired result.
\end{proof}

\begin{theorem} \label{thm:coefmag}
Let $B$ be a top set of size $k$.
Let $S = \{s_1,\dots,s_{k}\}$ be a subset of $[n]$.
Let $A \prec B$.
Then $\coef{\chi_{A,B}}{S} \neq 0$ if and only if (1) $S-B \subseteq A$ and (2) there is no pair $\alpha \in S-B$ and $\beta \in S \cap B$ such that the index of $\alpha$ in $A$ equals the index of $\beta$ in $B$. 
\end{theorem}
\begin{proof}
Suppose $\coef{\chi_{A,B}}{S} \neq 0$.
Then each $s_i$ must appear in either $A$ or $B$, since otherwise $x_{s_i}$ would appear nowhere in $\chi_{A,B}$.
Thus if $s_i \notin B$, it must be in $A$, proving condition (1).
Suppose that there is some pair $\alpha \in S-B$ and $\beta \in S \cap B$, where the index of $\alpha$ in $A$ is $j$ and the index of $\beta$ in $B$ is $j$.
Then $\chi_{A,B} = (x_{\alpha} - x_{\beta})P(x_1,\dots,x_n)$ where $P(x) = \prod_{i=1,i \neq j}^{d} (x_{A_i} - x_{B_i})$.
Now $P$ contains neither $x_{\alpha}$ nor $x_{\beta}$, since all elements of $A$ and $B$ are distinct from each other.
Then, neither $x_{\alpha}P(x_1,\dots,x_{n})$ nor $x_{\beta}P(x_1,\dots,x_{n})$ contain a term with both $x_\alpha$ and $x_\beta$, and so neither can their sum $\chi_{A,B}$.
This proves condition (2).

Now instead suppose that conditions (1) and (2) hold.
There are $d$ different factors $(x_{A_i} - x_{B_i})$ which divide $\chi_B$.
Since all $d$ elements of $S$ must appear as $A_i$ or $B_i$ for some $i$, and by condition (2) no two elements of $S$ can appear as $A_i$ and $B_i$ for the exact same $i$, it follows that for each $i$ exactly one of $A_i$ and $B_i$ is in $S$.
Thus by expanding $\chi_B$ into monomials, we see that $x_{s_1}\cdots x_{s_{k}}$ has a nonzero coefficient.
\end{proof}

\begin{theorem}\label{CuandoA}
For a top set $B_{n, d} = \{b_1,b_2,\dots,b_{d}\}$ there are $(b_1 - 1)(b_2 - 3)(b_3-5)\cdots(b_{d} - 2d + 1)$ sequences $A \in S_{n, d}$ such that $A \prec B$.
\end{theorem}
\begin{proof}
We prove the statement by induction. For $d=1$, there are $b_1-1$ numbers in $[n]$ less than $b_1$.
	For $d > 0$, we choose one of $(b_1-1)(b_2-3)\cdots(b_{d-1}-2(d-1)+1)$ possibilities for the first $d-1$ elements of $A$.
	The last element can by any of the $b_{d}-1$ elements in $[n]$ which are less than $b_{d}$ and not one of the $2(d-1)$ elements previously chosen for $A$ or taken by $B$.
Thus we can choose any of $b_d-1-2(d-1) = b_d-2d+1$ options for it, yielding our formula when multiplied by the number of ways to chose the first $d-1$ elements.
\end{proof}

We have proven that the magnitude and sign of the coefficients for $\chi_B$ can be derived by the coefficients of monomials for $\chi_{A,B}$.
We have also provided formulas for calculating $\chi_B$ when $|B|=k$. 
We will use these theorems in the following algorithm to generate coefficients of $\chi_B$ in the construction of the eigenvectors.

\section{Our Algorithm}
Next we describe an algorithm for generating an orthogonal eigenbasis, along with a procedure for projecting a vector onto each $M_i$ for $0 \leq i \leq d$.
This significantly reduces runtime and will prime us for calculating each $f_i$ by projecting $f$ onto each $M_i$. In order to project $f$ we describe a method of computing eigenvectors which uses the linear mapping defined in Theorem \ref{coef}.
Consider a top set $B$ of length $d$ and a subset $S \subseteq [n]$ of size $d$.
Let $s_1 < s_2 < \dots < s_{d}$ be the elements of $S$.
The procedure described in Algorithm 1 allows one to extract the coefficient of $x_{s_1}\cdots x_{s_{d}}$ in $\chi_B$.

\begin{algorithm}[H]
\caption{Extract coefficient of $x_{s_1}\cdots x_{s_{d}}$ in $\chi_B$.}
\begin{algorithmic}[1]
\Procedure{ExtractCoefficient}{$B, S$}
   \State $i \gets 1$
   \State $j \gets 0$
   \State $\mathit{answer} \gets 1$
   \While {$i \leq d$}
		\If{$j < d$ and $B_i = s_{j+1}$}
            \State $\mathit{answer} \gets \mathit{answer} \cdot (i + j - B_i)$
            \State $i \gets i+1$
            \State $j \gets j+1$
        \ElsIf{$j = d$ or $B_i < s_j$}
            \State $\mathit{answer} \gets \mathit{answer} \cdot (j - i + 1)$
            \State $i \gets i+1$
        \Else
            \State $j \gets j+1$
        \EndIf
   \EndWhile
   \State \textbf{return} $\mathit{answer}$
\EndProcedure
\end{algorithmic}\label{extract_coef_alg}
\end{algorithm}

\subsection{Proof of Correctness}
\begin{theorem}
Let $\ell(b)$ be the number of $s_i \in S$ which are less than $b$. Then $|\coef{\chi_B}{S}| = \prod_{i=1}^{d} W_i$ where $W_i = B_i - i - \ell(B_i)$ if $B_i \in S$ and $W_i = \ell(B_i) - i + 1$ otherwise.
\end{theorem}
\begin{proof}
By Theorem \ref{thm:coefmag}, $|\coef{\chi_B}{S}|$ is the number of $A \prec B$ such that (1) $S-B \subseteq A$ and (2) there is no pair $\alpha \in S - B$ and $\beta \in S\cap B$ such that the index of $\alpha$ in $A$ equals the index of $\beta$ in $B$.

We will walk through a process with steps $1,2,\dots, d$ for constructing such an $A$ where these conditions hold.
At step $i$, we fix a value for the $i$-th element in $A$.
If $B_i$ is in $S$, then $A_i$ must satisfy all of these conditions
\begin{itemize}
    \item[I.] $A_i < B_i$
    \item[II.] $A_i \notin S$
    \item[III.] $A_i \notin B$
    \item[IV.] $A_i \neq A_j$ for any $j < i$.
\end{itemize}
There are $B_i-1$ elements initially satisfying the first condition.
We then subtract out the $\ell(B_i)$ elements which do not satisfy the second condition.
Let $w$ be the number of elements in $B$ which are less than $B_i$ and not in $S \cap B$.
Then we subtract off the $w$ elements remaining which do not satisfy the third condition.
Finally, there are $i-1-w$ elements remaining which are in $A$, but not in $S$.
	In total there are $(B_i-1) - \ell(B_i) - (w) - (i-1-w) = B_i - i - \ell(i)$ choices for $A_i$.

Now suppose that $B_i$ is not in $S$.
Then $A_i$ must be in $S$, must be less than $B_i$, and must not be in $B$.
	Thus, there are $\ell(B_i) - (i-1) = \ell(B_i) - i + 1$ possible choices in this case.
This proves the theorem.
\end{proof}

Algorithm \ref{extract_coef_alg} uses a two-pointers method to iterate over all elements in $B \cap S$ in order from least to greatest.
At the start of each iteration of the while loop, $j = \ell(B_i)$.
When an element $B_i$ in $B \cap S$ is iterated over, we multiply $\mathit{answer}$ by $i + \ell(B_i) - B_i = -(B_i - i - \ell(B_i))$.
Multiplying by this factor accounts for both the combinatorial change in number of satisfying sets $A$, as well as the sign change incurred by increasing the size of $B \cap S$ by one.
When an element in $B_i$ in $B-S$ is iterated over, we multiply $\textit{answer}$ by $\ell(B_i)-i + 1$.
At the end of the algorithm, we have that
\[
    \mathit{answer} =
    (-1)^{|B \cap S|}\prod_{i}^{d} W_i =
    (-1)^{|B \cap S|} |\coef{\chi_B}{S}| =
    \coef{\chi_B}{S}
\]
by Theorem \ref{sgnCof}.
This proves correctness of the algorithm.

\subsection{Enumerating Top Sets}
Let 
\[
    \mathcal{B}^{n,k} = \bigcup_{i=d}^k \mathcal{B}_{n,d}.
\]
In this section, we detail generation of $\mathcal{B}^{n,k}$ using a recursive method, where every recursive call corresponds to a single element of $\mathcal{B}^{n,k}$.
The generation algorithm is described below.
We use \textproc{Append}($L$, $x$) to denote the appending of item $x$ to the list $L$.

\begin{algorithm}[H]
\caption{Recursively generate all of $\mathcal{B}^{n,k}$}
\begin{algorithmic}[1]
    \Procedure{GenerateB}{$B$, $v$, $n$, $k$}
    \State \Call{Print}{$B$}
    \If{$|B| = k$}
        \textbf{return}
    \EndIf
    \State $b \gets \min\{x \in [n] : \forall i \in [d], x > B_i \}$. \label{declare_b_line}
    \If{$v = 0$}
        \State $b = b+1$
    \EndIf
    \While {$b \leq n$}
        \State $q \gets 0$
        \If {$|B| > 0$}
            \State $q \gets B_{|B|}$
        \EndIf
        \State \Call {GenerateB}{\Call{Append}{$B$, $b$}, $b - q - 2$} \label{gen_rec_call_line}
        \State $b \gets b+1$
    \EndWhile  
\EndProcedure
\end{algorithmic}\label{top_set_enum_alg}
\end{algorithm}

We now prove correctness of the algorithm.

\begin{theorem}
    In \textproc{GenerateB}, at all times $v = q(B)-2|B|$, where $q(B) = |\{x \in [n] : x \leq b \text{ for some } b \in B\}|$
\end{theorem}
\begin{proof}
    The length of $B$ is only zero in the base call of the recursion, so trivially $v=q(\emptyset) = 0$ if $|B| = 0$.
    When we append a value on the end of $B$ to make a new sequence $B'$, then $v'$ must equal $q(B') - 2|B'|$. Thus
    \[
        v' - v = q(B') - 2(|B|+1) - (q(B) - 2|B|) = q(B') - q(B) - 2.
    \]
    To complete the proof, note that $q(B) = B_{|B|}$ when $|B| \neq 0$.
    The variable $b$ is always larger than the largest value in $B$, thus when it is appended $q($\textproc{Append}$(B, b)) = b$.
    This corresponds directly to the value of $v$ passed into \textproc{GenerateB} on line \ref{gen_rec_call_line}.
\end{proof}

\begin{theorem}
\textproc{GenerateB}($\emptyset$, 0, $n$, $k$,) enumerates all top sets in $\mathcal{B}^{n,k}$.
\end{theorem} \label{generateb_thm}
\begin{proof}
    We prove inductively that all $B \in \mathcal{B}_{n,d}$ are printed in calls to \textproc{GenerateB} of recursion depth $d$ exactly once (where the base call has recursion depth 0), and no other sets are printed.
    When $d=0$, the single base call at recursion depth 0 prints the null sequence, the single top set of length 0.
    
    Now let $d > 0$ and suppose all elements in $\mathcal{B}_{n,d-1}$ are printed exactly once in calls of recursion depth $d-1$.
    Choose any $B' \in \mathcal{B}_{n,d}$.
    Calls of depth $d$ can only print sequences of length $d$, since with each recursive call one element is appended to the end of the $B$ variable.
    As the $B$ variable is only appended to and never changed during the recursion, if $B'$ is printed it must be printed from a call whose $B$ variable is the prefix of $B'$ excluding only the last element.
    In the loop making recursive calls to \textproc{GenerateB}, the variable $b$ is different each time.
    Thus $B'$ cannot be printed more than once.
    
    We now prove that $B'$ is printed at least once.
    Let $B$ be the prefix of $B'$ excluding only the last element, and let $b$ be the new element added.
    If appending $b$ to $B$ creates a new top set in $\mathcal{B}_{n,d}$, then so does appending $b+1$ as long as $b+1 \leq n$, since any $A \prec$ \textproc{Append}($B$, $b$) also satisfies $A \prec $ \textproc{Append}($B$, $b+1$).
    Therefore in order to call \textproc{GenerateB} on all top sets, we need find the smallest $b$ such that \textproc{Append}($B$, $b$) is a top set.
    Clearly $b$ is larger than the all elements in $B$, and this is guaranteed on line \ref{declare_b_line}.
    For \textproc{Append}($B$, $b$) to be a top set it must be the true that $b - 2(|B|+1)\geq 0$.
    If some $A \prec$ \textproc{Append}($B$, $b$) then the total elements in $A$ and \textproc{Append}($B$, $b$) combined is $2(|B|+1)$, and all of these elements are at most $b$.
    Given $B$ is a top set, we can always construct some sequence $A \prec$ \textproc{Append}($B$, $b$) by extending any sequence $\prec B$ with any of the $b-2(|B|+1)-1$ elements not in $A$ or $B'$ so far.
    This proves correctness of the generation algorithm.
\end{proof}

All operations except \textproc{Print}, and the loop take constant time, and all operations within the loops except \textproc{GenerateB} take constant time as well.
All loop iterations fix a valid element of (at least) one valid top set, so the runtime to generate all of $\mathcal{B}^{n,k}$ is absorbed into the runtime of any processing we might perform which looks at all entries of all elements in $\mathcal{B}^{n,k}$.
Since all algorithms presented in this paper do so, we ignore the runtime contribution of generating coefficients for the remainder of the paper.

\subsection{Extracting Eigenvectors}
We now have a fast procedure for calculating each $\chi_B$.
All that remains is to describe a computational procedure for extracting the eigenvector associated to a given $\chi_B$.
To do this, we use the formula given in Theorem \ref{coef}.
It takes $O\left(k\binom{n}{k}\right)$ time to generate all coefficients of $\chi_B$.
Naively applying the linear transformation to some $\chi_B$ where $|B| = d$ takes time $O(\binom{n}{k} \binom{k}{d})$. Thus, the algorithm takes order
\begin{align*}
\sum_{d=0}^{k} \left(\binom{n}{d} - \binom{n}{d-1}\right) \left(\binom{n}{k}\binom{k}{d} + \binom{n}{d}d\right)
\end{align*}
time.
At a slight cost to parallelizability, we can reduce the sequential runtime further.
Let $L_{a,b}$ with $a < b$ denote the linear transformation from $\R^{\binom{[n]}{a}}$ to $\R^{\binom{[n]}{b}}$ defined by $L(f)_S = \sum_{T \subset S, |T| = a} f_T$.
Then we would like to apply the transformation $L_{d,k}$, and to do so we will apply $L_{i,i+1}$ for every $d \leq i < k$ in sequence and scale the result.
Note that $L_{j-1,j} L_{i, j-1} = (j-i)L(j,i)$, since there are $j-i$ ways to choose an intermediate subset of size $j-1$ between any size $j$ subset and some size $i$ subset of it.
Thus, the value at each subset is counted $j-i$ times in $L_{j-1,j} L_{i, j-1}$.
We may use this recurrence to compute $L_{d,k}$ in total time
\[\sum_{r=d+1}^k r\binom{n}{r}.\]
This brings the runtime of the algorithm to
\begin{align*}
\sum_{d=0}^k \left(\binom{n}{d} - \binom{n}{d-1}\right)\sum_{r=d+1}^k r \binom{n}{r}
&\leq \sum_{d=0}^k\sum_{r=d}^k \left(\binom{n}{d} - \binom{n}{d-1}\right) r \binom{n}{r} \\
&= \sum_{r=0}^k r \binom{n}{r} \sum_{d=0}^r \left(\binom{n}{d} - \binom{n}{d-1}\right) \\
&= \sum_{r=0}^k r \binom{n}{r}^2 \\
&\leq k^2 \binom{n}{k}^2. \\
\end{align*}
in total.

\subsection{Projecting onto Subspaces}
Fix $n$, $k \leq n/2$, and $d \leq k$.
Let $f \in \R^{\binom{n}{k}}$.
We will describe a computational procedure for projecting $f$ onto the subspace $M_d$.
To project $f$ onto $M_d$, it suffices to project onto an orthogonal set of basis vectors for $M_d$ and sum the projections \cite{uminsky2003generalized}.
Let $c_B$ be a vector containing all the coefficients of $\chi_B$.
Let $e_B$ be the eigenvector associated with $B$, namely containing the results of evaluating $\chi_B(x_0,\dots,x_{n-1})$ at all $x_0,\dots,x_{n-1}$ where exactly $k$ values of $x_i$ are 1 and the remaining are 0.
In Theorem \ref{coef}, it is proven that the map $c_B \mapsto e_B$ is a linear transformation $L$.
Recall that  for $B \in \mathcal{B}_{n,d}$, the set of $e_B$ forms an orthogonal basis for $M_d$.
Let $f_d$ be the projection of $f$ onto $M_d$.
Then
\begin{align*}
f_d
&= \sum_{B \in \mathcal{B}_{n,d}} \operatorname{proj(f, e_B)}
\\ &= \sum_{B \in \mathcal{B}_{n,d}} \frac{f^T e_B}{\|e_B\|^2}e_B
\\&= \sum_{B \in \mathcal{B}_{n,d}} \frac{f^T Lc_B}{\|e_B\|^2}Lc_B
\\&= L\sum_{B \in \mathcal{B}_{n,d}} \frac{(f^T L)c_B}{\|e_B\|^2}c_B.
\end{align*}

A theorem of Filmus \cite{filmus2016orthogonal} gives an explicit formula for $\|e_B\|^2$ under a different norm. However, the norm used by Filmus differs from the one used for Johnson Graphs by exactly a factor of $\binom{n}{k}$.
Putting this information together, we obtain a formula

\[\|e_B\|^2 = \binom{n}{k}\left(\prod_{i=1}^{d}(b_i-2i+2)(b_i-2i+3)\right) \frac{k^{\underline{d}}(n-k)^{\underline{d}}}{n^{\underline{2d}}}\]
where $e_B$ is a vector containing all the coefficients of $\chi_B$, $e_B$ is the eigenvector associated with $B$, and $a^{\underline{b}} = a!/(a-b)!$.

In order to compute the projection $f_d$, the following procedure suffices.
First, compute $f^TL$.
Let $v$ be the zero vector of dimension $\binom{n}{d}$.
Iterate over all $B$, and compute both $c_B$ and $\|e_B\|^2$.
Scale $c_B$ by $\frac{(f^T L)c_B}{\|e_B\|^2}$, using our stored row vector for $f^TL$.
Then, add $c_B$ to $v$.
Once all $B$ have been iterated over, we have found $f_d = Lv$.

\section{Computational Complexity}
\subsection{Runtime of Projection}
We now look at the runtime of computing $f_d$.
While more efficient numerical methods exist for projection onto each $M_d$, the following method is both exact (non-numerical) and more efficient than simply generating an eigenbasis in $O(k^2 \binom{n}{k}^2)$ and performing projections.
The linear transformation $L$ can be represented as a matrix with $\binom{n}{k}$ rows and $\binom{n}{d}$ columns.
Each column contains exactly $\binom{k}{d}$ 1s, with the remaining entries being 0.
Using sparse matrix multiplication methods, applying $L$ to a vector takes $\binom{n}{k}\binom{k}{d}$ time, but as previously discussed, $L$ may be computed in time $\sum_{r=d}^k r\binom{n}{r}$.
The size of $\mathcal{B}_{n,d}$ is $\binom{n}{d} - \binom{n}{d-1}$ (where $\binom{n}{-1}$ is defined to be 0), as it is exactly the dimension of $M_d$ \cite{brouwer2012distance}.
For each such $B$, we compute all $\binom{n}{d}$ coefficients, which takes $O(d)$ time each. It takes an additional $O(d)$ time to compute $\frac{(f^T L)c_B}{\|e_B\|^2}$.
Finally, we take the summation of all $\frac{(f^T L)c_B}{\|e_B\|^2}$ and in order $O\begin{pmatrix}\binom{n}{k}\binom{k}{d}\end{pmatrix}$ time compute $f_d$ by applying $L$ to it.
Putting this information together, the total runtime of computing $f_d$ is of order
\[\sum_{r=d}^k r\binom{n}{r} + \left[\binom{n}{d} - \binom{n}{d-1}\right]\binom{n}{d}.\]

We now analyze the runtime of computing all such projections. Note that since $f = f_0+f_1+\cdots+f_{k-1} + f_k$, we need only explicitly compute $f_d$ for $d < k$ to obtain $f_k$ by subtracting the other $f_d$ from $f$.
Then the total runtime is of order
\begin{align*}
\sum_{d = 0}^{k-1}\sum_{r=d}^k r\binom{n}{r} + \left[\binom{n}{d} - \binom{n}{d-1}\right]\binom{n}{d}
&= \sum_{d = 0}^{k-1}\sum_{r=d}^k r\binom{n}{r} + \sum_{d = 0}^{k-1}\left[\binom{n}{d} - \binom{n}{d-1}\right]\binom{n}{d}d
\\ &\leq \sum_{r=0}^k \sum_{d = 0}^{k-1} r\binom{n}{r} + \sum_{d = 0}^{k-1}\left[\binom{n}{d} - \binom{n}{d-1}\right]\binom{n}{d}d
\\ &\leq (k-1)\sum_{r=0}^k r\binom{n}{r} + \sum_{d = 0}^{k-1}\left[\binom{n}{d} - \binom{n}{d-1}\right]\binom{n}{d}d
\\ &\leq k^3 \binom{n}{k} + \binom{n}{k-1}k\left[\binom{n}{k-1} - \binom{n}{-1}\right]
\\ &\leq k^3 \binom{n}{k}+ \binom{n}{k-1}^2 k.
\end{align*}

\subsection{Comparison of Runtime with Filmus's Construction}
We will show that a naive implementation of Filmus's description the eigenbasis, described in Section \ref{ortho_constr_sec}, has worse complexity than the algorithm presented in this paper.
We do so by finding a lower bound $L(n, k)$ for the runtime of Filmus's formula, and showing the asymptotic bound 
$\frac{L(n, k)}{k^2 \binom{n}{k}} = \Omega((3/2)^k/k)$.
The bounds presented here are very loose, and in practice the difference between our algorithm and Filmus's is even greater.
We start by proving a useful technical lemma.
We use the notation $x!!$ to denote the double factorial, defined as
\[x!! = \prod_{i=0}^{\lceil x/2-1\rceil}(n - 2i).\]

\begin{lemma}
Let $S \subseteq [n]$ be of size $2k$.
Then there are exactly $(2k-1)!!$ pairs $(A,B)$ such that $A \prec B$, $B \in \mathcal{B}_{n,k}$, and $S$ is the collection of all elements in either $A$ or $B$.
\end{lemma} \label{pair_count_lem}
\begin{proof}
We prove the statement by induction on $|S|$.
Let $C(k)$ be the number of pairs $(A, B)$ such that $A \prec B$, $B \in \mathcal{B}_{n,k}$, and $S$ is the collection of all elements in either $A$ or $B$.
For $k=1$, the lemma says $C(1) = 1$, which is true since $|S| = 2$ and the larger element of $S$ must be in $B$ while the smaller element must be in $A$.
Now let $|S| = 2k$ where $k > 1$.
Let $s \in S$ be the largest element in $S$.
If $s = A_i$ for some $i$, then $B_i > s$, but $B_i \in S$ contradicts the maximality of $s$.
If $s = B_i$ for some $i < k$, then $B_{k} > B_i = s$, but $B_k \in S$ contradicts the maximality of $s$.
Therefore $B_k = s$.
There are $|S|-1 = 2k-1$ choices for $A_k$ at this point; take any of them.
Regardless of our choice for $A_k$, we now must construct $(A', B')$ out of the elements $S - \{A_k, B_k\}$ such that $A' \prec B'$, then append $A_k$ to $A'$ to obtain $A$ and append $B_k$ to $B'$ to obtain $B$.
Note that any such choice is valid; we are guaranteed $B_k > B_{k-1}$ since $B_k$ was maximal in $S$.
Thus by induction, $C(k) = (2k-1)C(k-1) = (2k-1)(2k-3)!! = (2k-1)!!$.
\end{proof}
\begin{corollary}
There are $\binom{n}{2k}(2k-1)!!$ pairs $(A,B)$ with $A \prec B$ and $B \in \mathcal{B}_{n,k}$.
\end{corollary} \label{pair_count_cor}
\begin{proof}
There are $\binom{n}{2k}$ subsets of $[n]$ of size $2k$.
For each $S \subseteq [n]$, Lemma \ref{pair_count_lem} states that there are $(2k-1)!!$ pairs $(A,B)$ with $A \prec B$ such that $S$ is the collection of all elements in $A$ and $B$.
Each pair $(A, B)$ with $A \prec B$ is associated to exactly one such set $S$.
Summing $(2k-1)!!$ over all $S$, we obtain the desired result.
\end{proof}

We are now ready to obtain a lower bound for the runtime of a naive implementation of Filmus's formulas for an orthogonal eigenbasis of $J(n, k)$.
In order to use Filmus's formulas, we must evaluate $\chi_{A,B}(x_1, \dots, x_n)$ over all $(x_1,\dots,x_n) \in \binom{[n]}{k}$, every $B \in \mathcal{B}^{n,k}$, and every $A \prec B$.
Since we are only attempting to find a lower bound, we may restrict $B$ to be in $\mathcal{B}_{n,k} \subset \mathcal{B}^{n,k}$ instead.
Evaluating $\chi_{A,B}$ at $(x_1,\dots,x_n) \in \binom{[n]}{k}$ takes at least $\Omega(k)$ time by looping through the product decomposition of $\chi_{A,B}$ as $\prod_{i=1}^d (A_i - B_i)$.
Each $\chi_{A,B}$ must be evaluated at $\binom{n}{k}$ different inputs.
The number of $\chi_{A,B}$ that must be evaluated is equal to the number of pairs $(A,B)$ such that $A \prec B$ and $B \in \mathcal{B}_{n,k}$, which by Corollary \ref{pair_count_cor} is equal to $\binom{n}{2k}(2k-1)!!$.
Thus the runtime of the naive implementation is $\Omega(L(n,k))$ where
\[L(n, k) = k \binom{n}{k} \binom{n}{2k}(2k-1)!!.\]

Recall that the complexity of our algorithm is $O\left(k^2 \binom{n}{k}^2\right)$.
We show that
\[\frac{L(n, k)}{k^2{\binom{n}{k}^2}} = \Omega((3/2)^k/k),\]
proving that our algorithm has better asymptotics than a naive implementation of Filmus's formulas.
\begin{align*}
\frac{L(n,k)}{k^2 \binom{n}{k}^2} &= \frac{\binom{n}{2k}(2k-1)!!}{k \binom {n}{k}} \\
&= \frac{k!(n-k)!(2k-1)!!}{k(2k)!(n-2k)!} \\
&= \frac{1}{k}\left(\prod_{r=0}^{k-1} \frac{1}{k+1+r}\right) \left(\prod_{r=0}^{k-1} n-k-r \right)\left(\prod_{r=0}^{k-1} 2r+1 \right) \\
&= \frac{1}{k}\left(\prod_{r=0}^{k-1} \frac{(n-k-r)(2r+1)}{k+1+r}\right) \\
&\geq \frac{1}{k}\left(\prod_{r=0}^{k-1} \frac{(n/2-r)(2r+1)}{n/2+1+r}\right) \\
&= \frac{1}{k}\left(\frac{(n/2-r)(2r+1)}{n/2+1+r}\right)^k.
\end{align*}
Define $f_n(r) = \frac{(n/2-r)(2r+1)}{n/2+1+r}$, and let $n \geq 8$ and $k \geq 4$.
Note that for $n \geq 2$, $f_n(0) = \frac{n/2}{n/2+1} \geq 3$ and $f_n(n/2-1) = \frac{n-1}{n} \geq \frac{1}{2}$.
Let $x = n/2-r$ and $y = 2r+1$.
Then we can write $f_n(r) = \frac{xy}{x+y}$.
If $0 < r < n/2-1$, then $x$ and $y$ are both at least 2.
The minimum value of $\frac{xy}{x+y}$ given $x,y \geq 2$ is $\frac{2 \cdot 2}{2+2} = 1$, since the derivative is positive at every coordinate.
If $1 < r < n/2 - 2$, then $x$ and $y$ are at least 3, so $f_n(r) \geq \frac{9}{6} = \frac{3}{2}$ by the same argument.
Then we have
\[\frac{1}{k}\left(\frac{(n/2-r)(2r+1)}{n/2+1+r}\right)^k \geq \frac{1}{k} \cdot 3 \cdot \frac{1}{2} \cdot 1 \cdot 1 \cdot \left(\frac{3}{2}\right)^{k-4} = \frac{3(\frac{3}{2})^{k-4}}{2k}.\]
Since the asymptotics depend only on $k$, we also allow $k \leq 3$, as it can be masked by the constant factor. 
This proves that a naive implementation of Filmus's algorithm takes almost exponentially longer in $k$ than the algorithm presented in this paper.

\subsection{Comparison with Iterative Algorithms}
We provide a non-iterative algorithm for calculating eigenvectors of the Johnson association scheme.
Standard algorithms for calculating eigenvectors of both general and real symmetric matrices require iteration toward some degree of accuracy.
As a result, the runtime of our algorithm is not directly comparable to traditional algorithms because we lack the additional component which specifies the degree of precision desired. 

We compare our runtime to the runtime for the power iteration.
This algorithm is the `state-of-the-art' method for computing eigenpairs for general matrices.
For power iteration the runtime per step is $O(n^2)$ with a linear convergence rate \cite{strang2006linear}. 
This does not include the construction of the matrix, and only considers iterations from some initial guess for eigenvalues.

\section{Parallelization}
The computation of the eigenbases for $J(n, k)$ and the projection of $f$ onto the subsequent eigenspaces is highly parallelizable.
To demonstrate this, we provide a parallelized version of both computations and analyze their runtime assuming an unbounded number of processors.
We first consider the eigenbasis extraction algorithm for $J(n, k)$.

\begin{theorem} \label{parallel_generation_thm}
    An orthogonal eigenbasis for $J(n, k)$ can be generated in $O(n)$ time given unlimited parallel processors.
\end{theorem}
\begin{proof}
Recall that each eigenvector is described by a top set $B$ of length $d$ for $0 \leq d \leq k$.
Modify Algorithm \ref{top_set_enum_alg} so that each recursive call runs on a new processor.
The sequential runtime is the largest amount of time it takes generate any specific $B \in \mathcal{B}^{n,k}$.
Suppose we are in the body of \textproc{GenerateB} printing a top set.
We bound the amount of time it took to reach this point neglecting the time of recursive calls not currently on the stack, giving us our recursive runtime.
Note that, tracing only these recursive calls leading us to our current state, the value of the variable $b$ is always increasing.
In fact, each iteration of the inner while loop implies $b$ is increased by one, so in all, the inner loop is executed at most $n$ times to reach $B$.
There is $O(k)$ overhead from the function calls and thus generating all of $\mathcal{B}^{n,k}$ takes $O(n+k) = O(n)$ time in total.

Once all $B \in \mathcal{B}^{n,k}$ are generated, we need to generate subsets labelling the coefficients of $\chi_B$.
This is possible to do with a recursive search procedure in $O(n)$ time, where the search procedure passes down the current subset being built and a potential next element.
Details are omitted.
Once all subsets are generated, we run \textproc{ExtractCoefficient}($B$, $S$) on every pair, taking time $O(k)$.

So far, we have taken $O(n + n + k) = O(n)$ time to generate coefficients of all $\chi_B$.
All that remains is to compute the linear transformation $L$.
We compute each eigenvector entry in parallel, so take $O(n)$ time to generate all size $k$ subsets for each $B$.
The eigenvector entry associated with $S \subset n$ is a sum of the coefficients of $\chi_B$ associated to each size $|B|$ subset of $S$.
We can generate all subsets of $S$ in parallel in $O(|S|) = O(k)$ time, and sum the subsets of size $|B|$ as the threads halt and return in $O(k)$ time as well.
This procedure halts in $O(n)$ time with all eigenvectors generated.
\end{proof}

\begin{theorem}
    Given $f \in \mathbb{R}^{\binom{n}{k}}$, one can compute $\operatorname{proj}(f, M_d)$ for each $0 \leq d  \leq k$ in $O(n)$ time given unlimited parallel processors.
\end{theorem}
\begin{proof}
    We begin with the procedure described in Theorem \ref{parallel_generation_thm}, but modify it so that directly after computing each eigenvector $e_B$, it also projects $f$ onto $e_B$.
    It takes $O(n)$ time to find $\frac{|f \cdot e_B|}{|e_B|}e_B$, since dot products and scaling vectors takes $O(\log \binom{n}{k}) = O(\log (2^n)) = O(n)$ time.
    Recall that $\mathcal{B}^{n,k}$ is generated via a recursive search, in which all top sets of length $B$ are enumerated at recursion depth $d$ by the proof of Theorem \ref{generateb_thm}.
    Once projections onto each basis vector have been generated, the projections onto each subspace can be added together.
    Each vector entry can be added in parallel time $O(\log \binom{n}{k})$, and it takes $O(\log \binom{n}{k})$ to start all of these threads, so this final phase also takes $O(n)$ time.
\end{proof}

\section{Acknowledgements}
This paper is dedicated to our advisors, Dr. David Uminsky and Dr. Mario Banuelos. The authors were supported in part by NSF Grant \#DMS-1659138, NSA Grant \#H98230-18-1-0008, and Sloan Grant \#G-2017-9876.

\newpage
\printbibliography
\end{document}